\documentclass{article}

\usepackage{arxiv}

\usepackage[utf8]{inputenc} 
\usepackage[T1]{fontenc}    
\usepackage{hyperref}       
\usepackage{url}            
\usepackage{booktabs}       
\usepackage{amsfonts}       
\usepackage{nicefrac}       
\usepackage[table, svgnames, dvipsnames]{xcolor}
\usepackage{microtype}      
\usepackage{tikz}
\usepackage{amsfonts}
\usetikzlibrary{quantikz}
\usepackage{bbm}
\usepackage{bm}
\usepackage{amsthm}
\usepackage{amsmath, color}
\usepackage[ruled]{algorithm2e}
\usepackage{mathtools}
\DeclarePairedDelimiter\ceil{\lceil}{\rceil}
\DeclarePairedDelimiter\floor{\lfloor}{\rfloor}

\newtheorem{thm}{Theorem}[section]
\newtheorem{lem}[thm]{Lemma}

\newcommand{\Tr}{\mathtt{Tr}}

\title{SWAP Test for an Arbitrary Number of Quantum
States}

\author{
  Xavier Gitiaux \\
  Department of Computer Science\\
  Quantum Science and Engineering Center\\
  George Mason University\\
  Fairfax, VA 22030 \\
  \texttt{xgitiaux@gmu.edu} \\
   \And
 Ian Morris\\
  Department of Physics and Astronomy\\
    Quantum Science and Engineering Center\\
  George Mason University\\
  Fairfax, VA 22030 \\
  \texttt{imorris@gmu.edu} \\
   \And
 Maria Emelianenko\\
  Department of Mathematical Sciences\\
    Quantum Science and Engineering Center\\
  George Mason University\\
  Fairfax, VA 22030 \\
  \texttt{memelian@gmu.edu} \\
  \And
 Mingzhen Tian\\
   Department of Physics and Astronomy\\
     Quantum Science and Engineering Center\\
  George Mason University\\
  Fairfax, VA 22030 \\
  \texttt{mtian1@gmu.edu} \\ 
}

\begin{document}
\maketitle

\begin{abstract}
We develop a recursive algorithm to generalize the quantum SWAP test for an arbitrary number $m$ of quantum states requiring $O(m)$ controlled-swap (CSWAP) gates and $O(\log m)$ ancillary qubits.
We construct a quantum circuit able to simultaneously measure overlaps $|\langle \phi_i, \phi_j\rangle |^2$ of $m$ arbitrary pure states $|\phi_1\ldots \phi_m\rangle$. Our construction relies on a pairing unitary that generates a superposition state where every pair of input states is labelled by a basis state formed by the ancillaries. 

\end{abstract}

\keywords{}

\section{Introduction}

The quantum community has made significant advances in developing algorithms that hold the promise of speedup over their classical counterparts. However, implementing these algorithms into quantum circuits remains an open question. 
Many groups resort to machine learning \cite{cincio2018learning, fosel2021quantum}, genetic algorithms \cite{lee07,rasconi2019an}, linear algebra techniques \cite{re:amy:phds,nam2018automated} to design circuits that implement either exactly or approximately quantum primitives that are instrumental to many quantum algorithms. 
Researchers are also developing optimization methods to efficiently map logical circuits to physical quantum machines \cite{itoko2020optimization,zulehner2018an}. 

In our work, we tackle the problem of generalizing the SWAP test to include an arbitrary number of overlapping states. Two-state SWAP tests are useful primitives in quantum computing which rely on controlled-swap (CSWAP) gates and an ancillary qubit to estimate the inner product $|\langle \phi, \psi\rangle|^{2}$ between two pure quantum states $\phi$ and $\psi$ \cite{kopczyk2018quantum}, as well to compute the overlap $\Tr(\rho \sigma)$ of two mixed states $\rho$ and $\sigma$ \cite{garcia2013swap}. If $\rho=\sigma$, a SWAP test measures the purity of the quantum state $\sigma$ \cite{garcia2013swap}. For two $q$-qubit states $\phi,\psi$, the cost of evaluating $|\langle \phi, \psi\rangle|^{2}$ on a classical computer grows exponentially with $q$, while the depth of the SWAP test circuit on a quantum machine grows linearly with $q$ \cite{garcia2013swap}, producing an exponential speedup that brings a quantum advantage provided that error correction/mitigation is achievable.
 
SWAP test appears in a variety of applications that depend on estimation of state overlaps. It has been used in quantum state characterization, such as measurement of entanglement and indistinguishabiliy \cite{Foulds2021, garcia2013swap}. It is useful for benchmarking on a quantum computer by monitoring the amount of decoherence that occurs through a circuit. It is also a key module in hybrid quantum machine learning algorithms.
For instance, in quantum supervised learning \cite{kopczyk2018quantum,ciliberto2018quantum,havlicek2019supervised} and in particular in
quantum support vector machines (SVM) \cite{rebentrost2014quantum, rudolph2020generation}, it is used to measure pairwise distances at the stage of cluster assignment.
 
 While the question of improving the SWAP test for two high-dimensional states has been subject to several studies \cite{cincio2018learning, garcia2013swap}, to date, with the exception of a 3-state model presented in \cite{Hai-Rui, Galvao2019}, SWAP tests have been only implemented between two quantum states at a time \cite{kathuria2020implementation}. 
 
 We study how to extend the SWAP test to $m$ arbitrary pure quantum states $\ket{\phi_{1}...\phi_{m}}$ and compute $|\langle \phi_{i}, \phi_{j}\rangle|^{2}$ for all pairs $i\neq j$ using at most $O(\log{m})$ measurements. Our approach relies on a pairing unitary $\mathcal{U}_{m}$ made of $O(m)$ CSWAP gates controlled by $O(\log{m})$ ancillary qubits. The pairing operator applied to $m$ state registers outputs a superposition state where all pairs $(\phi_{i}, \phi_{j})$ appear with positive probability on the first two state registers and are labelled with a basis state formed by the ancillaries. By applying a standard two-state SWAP test to the resulting superposition state, measurement on the ancillary qubits matches the probability of an ancillary basis state to the state overlap this ancillary basis state is assigned to label.
    
First, we construct a unitary $\mathcal{U}_{4}$ for a 4-state input with three ancillaries and three CSWAP gates. We argue that the resulting circuit is optimal with respect to the number of CSWAP gates when limiting the number of measurements to four. Then, we use our 4-state circuit as a building block to develop a recursive algorithm capable of producing an $m$-state circuit for which all overlaps $|\langle \phi_{i}, \phi_{j}\rangle|^{2}$ are measured on $O(\log{m})$ ancillary qubits. We give a formal analysis of its complexity and show that the resulting circuit relies on $c_m=O(m)$ CSWAP gates and $d_m=O(\log m)$ ancillaries. 

Since our estimates of state overlaps are probabilistic, we need to repeat our experiment to obtain an estimate of each state overlap with a mean squared error less than some precision $\epsilon$. We model these repeated experiments by assuming access to an oracle that constructs the states $\ket{\phi_{1}}, ..., \ket{\phi_{m}}$. We argue that for arbitrary $m$, the constraint to limit the number of measurements to $O(\log m)$ is necessary to guarantee a polynomial number of calls to the oracle.  

Numerically, we implement a genetic procedure to search how to optimally apply CSWAP gates to maximize the number of pairs that appear on the two first registers of a quantum circuit. Our optimization procedure is motivated by the fact that genetic algorithms are well
suited for optimization problems with complex parameter and energy landscapes. Genetic algorithms have been used in a variety of applications including scheduling, building design, trajectory optimization \cite{Hart, Yu}. They are also successful in automating classical circuit design \cite{aly} and more recently, in designing quantum circuits \cite{Potocek}. For $m=8$ input states and six measurements, we show empirically that a genetic algorithm does not find circuits with fewer than nine CSWAP gates, which is the number of CSWAP gates we use in our recursive construction (see \ref{circuit: U3}). 


\begin{figure}[t]
    \centering
    \tikzset{every picture/.style={scale=3}}
    \input{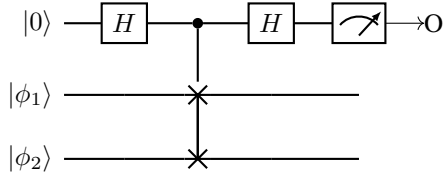}
    \caption{Quantum circuit implementing the SWAP test for two quantum states.}
    \label{fig: circuit1}
\end{figure}

The paper is organized as follows. We present our main algorithm in Section \ref{sec: alg}, which includes a pairing unitary for four states, then, a recursive formula for $m=2^{k}$ states with exact gate and ancillary counts, and a generalization to arbitrary $m$ states. In Section \ref{sec: num}, we provide the results of a genetic algorithm numerical search that gives empirical evidence for our arrangement of CSWAP gates to label each pair of inputs is optimal at least for $m=8$. Possible extensions and open questions are discussed in Section \ref{sec: discussion}.

\section{Multi-state SWAP test algorithm}
\label{sec: alg}
We consider $m$ quantum states $\ket{\phi_{i}}$ for $i=1,\ldots,m$, with $m\in \mathbbm{Z}$. The objective is to obtain the inner products $|\langle \phi_{i}|\phi_{j}\rangle|^{2}$ for $i\neq j$. 

\subsection{2-State Algorithm}

 The quantum circuit that operates the swap test (Figure \ref{fig: circuit1}) includes two registers carrying the states $\ket{\phi_{1}}$ and $\ket{\phi_{2}}$ and an ancillary in $\ket{0}$. The state registers can be  multi-qubit or qudit in general. However, we will treat them as qubits for simplicity. The ancillary qubit is independent of the dimension of a state register. The first Hadamard gate prepares the ancillary from $\ket{0}$ into a superposition state $\ket{+}=\frac{\ket{0} + \ket{1}}{\sqrt{2}}$. The CSWAP gate exchanges the pair of input states $\ket{\phi_{1}}$ and $\ket{\phi_{2}}$ provided that the ancillary qubit is in state $\ket{1}$. The last Hadamard gate completes the circuit that transforms the input pair $\ket{\phi_{1}}\ket{\phi_{2}}$ into a superposition of symmetric and anti-symmetric states coupled with the ancillary in $\ket{0}$ and $\ket{1}$, respectively:
\begin{equation}
  \ket{0}\ket{\phi_{1}}\ket{\phi_{2}} \rightarrow \frac{1}{2}\ket{0}\left(\ket{\phi_{1}}\ket{\phi_{2}} + \ket{\phi_{2}}\ket{\phi_{1}}\right) +  \frac{1}{2}\ket{1}\left(\ket{\phi_{1}}\ket{\phi_{2}} - \ket{\phi_{2}}\ket{\phi_{1}}\right). 
\end{equation}
The probability $p_{0}$ of measuring the ancillary in $\ket{0}$ relates to the inner product $|\langle\phi_{1}|\phi_{2}\rangle|^{2}$.
\begin{equation}
    p_{0} = \frac{1 + |\langle \phi_{1}|\phi_{2}\rangle|^{2}}{2}.
\end{equation}
Therefore, the two-state swap test is equivalent from a measurement perspective to the following transformation:
\begin{equation}\nonumber
    \ket{0}\ket{\phi_{1}}\ket{\phi_{2}}\rightarrow \sqrt{p_{0}}\ket{0y} + \sqrt{1 - p_{0}}\ket{1y'}
\end{equation}
 where $y$, $y'$ are garbage states.

This simple two-state SWAP can be used to handle $m$ quantum states by randomly pairing the input states and applying the circuit in Figure \ref{fig: circuit1} to each pair $\ket{\phi_{i}}, \ket{\phi_{j}}$. However, there are two limitations to this approach. First, it requires a polynomial number of CSWAP gates, one for each pair. Second, the two arbitrary input states cannot be recovered from the output (either symmetric or anti-symmetric) state after measurement on the ancillary \cite{garcia2013swap}. Therefore, in order to implement pairwise inner product calculations on $m$ states, $O(m^{2})$ copies of each state is needed.

Our goal is to build a circuit that can take all $m$ input states at once. The two-state circuit in Figure \ref{fig: circuit1} can still be used to compute the inner products if the two top registers are prepared in a superposition of all $m(m-1)/2$ pairs and each pair is labelled by an orthogonal state of some extra ancillary qubit. We will focus on an algorithm that constructs a pairing unitary $\mathcal{U}_{m}$ that generates this superposition state with proper labelling for $m=2^{k}$ states in section \ref{sec: 2k}; and, for any $m$ in section \ref{sec: any}. In section \ref{sec: th}, we show how a pairing unitary followed by a standard SWAP test relates measurements to state overlaps $|\langle \phi_{i}, \phi_{j}\rangle|^{2}$ for all $i\neq j$.

\subsection{Paring unitary $\mathcal{U}_{m}$ construction for $m=2^k$}
\label{sec: 2k}

First, we construct the pairing unitary $\mathcal{U}_{4}$ for four input states (see Figure \ref{circuit: U2}). There are a total of six pairs of input states, which requires at least 6 orthogonal states to label them. Therefore, the minimum number of ancillary qubits is three in order to encode at least six pairs. The minimum number of CSWAP gates is also three, each ancillary qubit acting as control once. This circuit is the pairing unitary $\mathcal{U}_{4}$ with the smallest CSWAP gates and ancillaries. 

The circuit in Figure \ref{circuit: U2} generates a superposition state containing the terms listed in Table \ref{tab: u2}, where $s_{1,2,3}$ denote the basis states of the three ancillary qubits numbered from bottom to top in the ancillary section in Figure \ref{circuit: U2}. The four state registers are arranged from left to right in the table where indices of the input states  $i, j, i', j' =1, 2, 3, 4$, respectively. The input state $\ket{\phi_{1}\phi_{2}\phi_{3}\phi_{4}}$ matches with the ancillary basis state ${\ket{000}}$ which does not trigger any swap. For all other ancillary basis states, every qubit in ${\ket{1}}$ imparts a swap between the states specified by the corresponding CSWAP gate and thus reorders the input states in the four registers. The transformation by $\mathcal{U}_{4}$ brings every possible pair of the four input states to the first two registers and matches each pair with a different ancillary basis state. 
Labelling is redundant since there are more ancillary basis states than state pairs. In Table \ref{tab: u2} pairs $(\phi_{i},\phi_{i'})$ and $(\phi_{i'}, \phi_{i})$ are matched to both labelling states $\ket{001}$ and $\ket{101}$.         

  \begin{figure}[h]
    \centering
    \tikzset{every picture/.style={scale=3}}
    \input{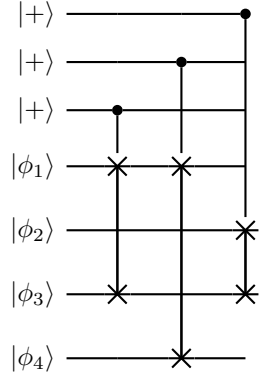}
    \caption{Quantum circuit implementing a pairing unitary $\mathcal{U}_{4}$ for four input states $\ket{\phi_{1}}, \ket{\phi_{2}}, \ket{\phi_{3}}, \ket{\phi_{4}}$.}
    \label{circuit: U2}
\end{figure}

\begin{table}[h]
    \centering
   \begin{center}
\begin{tabular}{ |c|c|c|c|c| }
\hline
 $s_1s_2$& \ket{00} & \ket{01} & \ket{10} & \ket{11} \\ 
 \hline
  $s_{3}=\ket{0}$ & $\phi_{i}\phi_{j}\phi_{i'}\phi_{j'}$ & $\phi_{j^{'}}\phi_{j}\phi_{i'}\phi_{j}$ & $\phi_{i'}\phi_{j}\phi_{i}\phi_{j'}$&$\phi_{j^{'}}\phi_{j}\phi_{i}\phi_{i^{'}}$\\ 
 \hline
 $s_{3}=\ket{1}$ & $\phi_{i}\phi_{i'}\phi_{j}\phi_{j'}$ & $\phi_{j'}\phi_{i'}\phi_{j}\phi_{i}$ & $\phi_{i'}\phi_{i}\phi_{j}\phi_{j'}$&$\phi_{j'}\phi_{i}\phi_{j}\phi_{i'}$\\  
 \hline
 \end{tabular}
\end{center}
    \caption{State of registers $1, 2, 3, 4 $ generated by $\mathcal{U}_{4}$, where the input state indices are $i, j, i', j' = 1, 2, 3, 4$ respectively.}
    \label{tab: u2}
\end{table}
 
In  Figure \ref{circuit: U3}, we construct  $\mathcal{U}_{8}$ from $\mathcal{U}_{4}$ to label the 28 pairs among $8$ inputs. The circuit is built into two sections of three ancillaries and three groups of three CSWAP gates, which amounts to six ancillaries and nine CSWAP gates. The input registers (1 through 8) are set into two groups of four (1 through 4 and 5 through 8). The first two groups of CSWAP gates  are controlled by the same set of ancillaries and apply $\mathcal{U}_{4}$ to both register groups in parallel. This brings all six pairs in $\ket{\phi_{1}}$ through $\ket{\phi_{4}}$ to registers $1$ and $2$ and all six pairs in $\ket{\phi_{5}}$ through $\ket{\phi_{8}}$ to registers $5$ and $6$. 
The last group of CSWAP gates controlled by another set of three ancillaries imparts an $\mathcal{U}_{4}$ to the four registers $1, 2, 5$, and $6$, which brings all 28 pairs of the eight input states to the first two registers ($1$ and $2$). Each pair is matched with a basis state generated by the six ancillaries. 
The $\mathcal{U}_{4}$ transformation tabulated in Table \ref{tab: u2} can be applied three times to the corresponding registers and ancillaries in Figure \ref{circuit: U3} to obtain the explicit superposition state at the output of $\mathcal{U}_{8}$.

\begin{figure}[t]
    \centering
    \tikzset{every picture/.style={scale=3}}
    \input{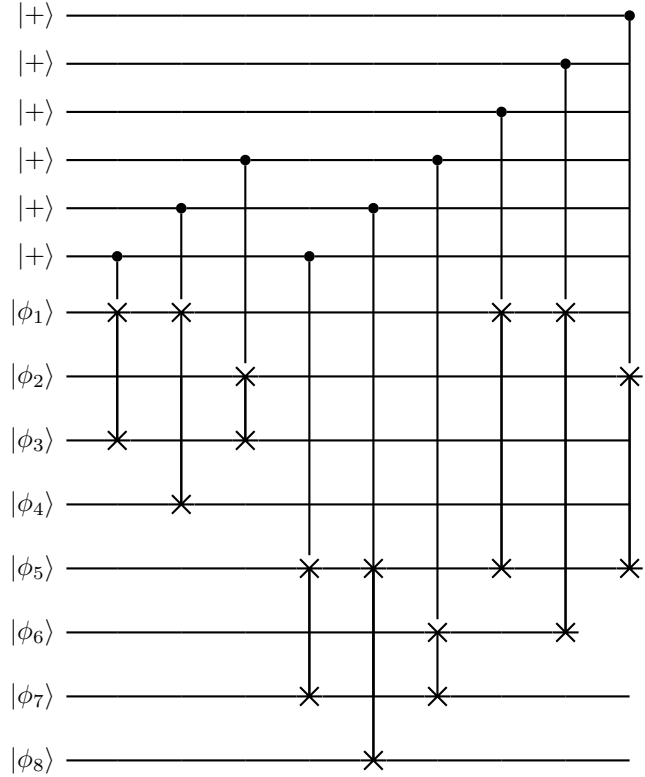}
    \caption{Quantum circuit implementing a pairing unitary $\mathcal{U}_{8}$ for $m=8$ inputs. The circuit stacks two pairing unitaries $\mathcal{U}_{4}$ (see Figure \ref{circuit: U2}) that share the same ancillaries followed by an additional pairing unitary $\mathcal{U}_{4}$ that acts on the first, second, fifth and sixth registers and is controlled by three additional ancillaries.}
     \label{circuit: U3}
\end{figure}

Based on $\mathcal{U}_{4}$ and the recursive method introduced to construct $\mathcal{U}_{8}$, we can scale up the scheme to $\mathcal{U}_{m}$ for arbitrary number $m=2^{k}$ of input states. $\mathcal{U}_{16}$ is built by repeating $\mathcal{U}_{8}$ controlled by the same set of six ancillaries plus an $\mathcal{U}_{4}$ at the end controlled by three more ancillaries. $\mathcal{U}_{m}$ circuit is constructed with $k-2$ recursive steps starting from $\mathcal{U}_{4}$. Therefore, if $c_{m}$ denotes the count of CSWAP and $d_{m}$ the number of ancillaries of $\mathcal{U}_{m}$, we have that $c_{m}=2c_{m/2} + 3$ and $d_{m}= d_{m/2} + 3$, that is $c_{m}=3/2 m -1$ and $d_{m}=3\log_{2}(m/2)$.


It should be noted that the operator $\mathcal{U}_{m}$ can be used to process any arbitrary number of input states for $2^{k-1}<m\le{2^{k}}$. We can use $\mathcal{U}_{2^{k}}$ after padding the extra-registers with known states (e.g. $\ket{0}$). It leads to a total of $3(k-1)$ ancillaries and $3\cdot 2^{k-1}-1$ CSWAP gates. 
Therefore, in the worst case $m=2^{k-1}+1$, by padding the input state, the complexity of the pairing operator $\mathcal{U}_{m}$ to process arbitrary $m$ input states is at most $3m + 2$ CSWAP gates and $2\log_{2}(m+1)$ ancillaries. The next section shows slightly tigher bounds by decomposing $m$ into $\floor{m/2}$ and $m - \floor{m/2}$ instead of padding $m$, but the order of magnitude for large $m$ remains the same.

\subsection{Pairing unitary and complexity estimation for arbitrary number of states}
\label{sec: any}

We generalize the recursive construction of the paring unitary $\mathcal{U}_{m}$ for an arbitrary $m$ input states. The transformation of $\mathcal{U}_{m}$ should take arbitrary unknown states $\ket{\phi_{1}}, ..., \ket{\phi_{m}}$ as input and indexes each pair $(\phi_{i}, \phi_{j})$ with a labelling state $\ket{ij}$ provided by the ancillaries as stated by Lemma 2.1.

  \begin{figure}
    \centering
    \tikzset{every picture/.style={scale=3}}
    \input{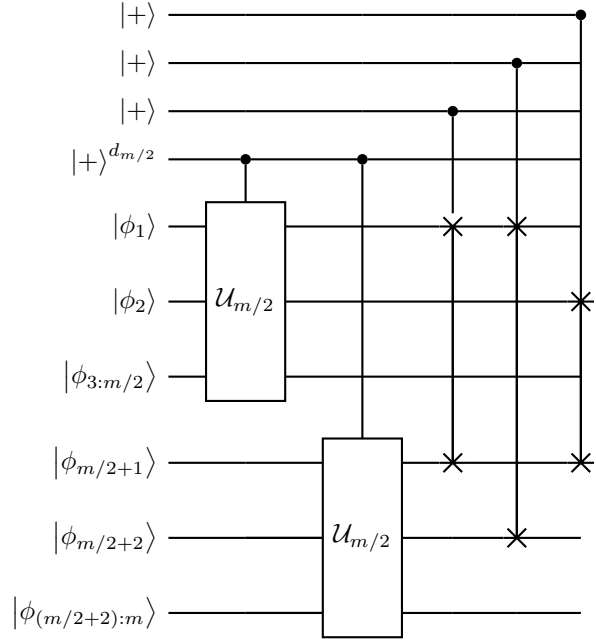}
    \caption{Generic construction of $\mathcal{U}_{m}$ from $\mathcal{U}_{m/2}$. $d_{m/2}$ denotes the number of ancillaries necessary to compute $\mathcal{U}_{m}$. The first pairing unitary $\mathcal{U}_{m/2}$ operates on states $\ket{\phi_{1}}, ..., \ket{\phi_{m/2}}$; the second pairing unitary on states $\ket{\phi_{m/2 + 1}}, ..., \ket{\phi_{m}}$.}
    \label{circuit: Uk}
\end{figure}

\begin{lem}
\label{lem: 1}
Let $n=m(m-1)/2$. Then, there exists a unitary $\mathcal{U}_{m}$ that maps
 \begin{equation}
    \mathcal{U}_{m}: \displaystyle\otimes\ket{0}^{d_{m}}\displaystyle\otimes_{i=1}^{m}\ket{\phi_{i}} \rightarrow \displaystyle\sum_{ij=1}^{n}\ket{ij}\ket{\phi_{i}\phi_{j}G_{ij}},
      \label{eq: pair}
\end{equation} 
where $G_{ij}$ is a garbage state stored in the state registers after the transformation and $d_{m}$ is the number of ancillaries. $\ket{ij}$ denotes the basis state of the $d_{m}$ ancillaries assigned to the state pair $\ket{\phi_{i}\phi_{j}}$. Moreover, $d_{m}=O(\ln m)$ and $\mathcal{U}_{m}$ can be computed by a quantum circuit with $O(m)$ CSWAP gates. 
\end{lem}

\paragraph{Constructive proof of Lemma \ref{lem: 1}.} The proof of Lemma \ref{lem: 1} relies on a recursive construction of $\mathcal{U}_{m}$ from $\mathcal{U}_{m'}$ for $m'<m$. 
\begin{equation}
    \mathcal{U}_{m^{'}}: \displaystyle\otimes\ket{0}^{d_{m'}}\otimes_{i-1}^{m'}\ket{\phi_{i}}\rightarrow \displaystyle\sum_{ij=1}^{n'}\ket{ij}\ket{\phi_{i}\phi_{j}G_{ij}},
\end{equation}

The $m$ input states can be set into two groups $\ket{\phi_{1}}$ through $\ket{\phi_{\ceil{m/2}}}$ and $\ket{\phi_{\ceil{m/2}+1}}$ through $\ket{\phi_{m}}$ as shown in Figure \ref{circuit: Uk}. An $\mathcal{U}_{\ceil{m/2}}$ and $\mathcal{U}_{\floor {m/2}}$ applied to the two groups, respectively, bring all pairs in the first group to registers 1 and 2 and all pairs in the second group to registers $\ceil{m/2}+1$ and $\ceil{m/2}+2$. At this point, an $\mathcal{U}_{4}$ applied to registers 1, 2, $\ceil{m/2}+1$, and $\ceil{m/2}+2$ completes the pairing of all states and bring them to registers 1 and 2. The labelling is provided by the basis states of $d_{\ceil{m/2}}+3$ ancillaries. The explicit pairing follows the same rule in Table \ref{tab: u2} with $i, j\in[1,..., \ceil{m/2}]$ and $i', j'\in[\ceil{m/2}+1,\ldots, m]$.

To analyze the complexity of $\mathcal{U}_{m}$, we notice that the count of CSWAP gates $c_{m}$ follows the recursive relation
\begin{equation}
\label{eq: rec}
    c_{m} = c_{\ceil{m/2}} + c_{\floor{m/2}} + 3.
\end{equation}
were the first two terms come from $\mathcal{U}_{\ceil{m/2}}$ and $\mathcal{U}_{\floor {m/2}}$, and the last term from the last $\mathcal{U}_{4}$.
Similarly, the total number of ancillaries used to compute $\mathcal{U}_{m}$ satisfies:
\begin{equation}
    d_{m} = d_{\ceil{m/2}} + 3,
\end{equation}
where $\mathcal{U}_{\floor {m/2}}$ uses the same of ancillaries as $\mathcal{U}_{\ceil{m/2}}$ and the three additional ancillaries control the last $\mathcal{U}_{4}$.

Assume that for all $m'<m$, $c_{m'}\le \kappa_1 {m'}-\kappa_3$ for $\kappa_{1}\geq 3/2$ and $\kappa_{3}\geq 3$. Then 
$c_{m}\le \kappa_1 \ceil{m/2} + \kappa_1 \floor{m/2} - 2\kappa_3+ 3 = \kappa_1 m -2\kappa_3+3$.
For $\kappa_3\ge 3$, we get $c_m\le \kappa_1 m - \kappa_3$, which is our inductive hypothesis. The hypothesis is also true for $m=4$ since $c_4=3$ and $d_4=3$. 
Therefore the CSWAP gate count is upper-bounded by $c_{m}=O(m)$.\\
We analyze the ancillary count in a similar way by assuming $d_{m'}\le \kappa_2 \log_{2}{(m'/2)}$ for all $m'<m$. Then, we have
\begin{equation}
\label{eq: dm}
\begin{split}
    d_{m} & \leq \kappa_{2}\log(\frac{m+1}{4}) + 3  \\
    & = \kappa_{2}\log(\frac{m}{2})-(\kappa_{2}\log(\frac{2m}{m+1}) - 3).
    \end{split}
\end{equation}
Therefore, $d_{m}\le \kappa_2 \log{(m/2)}$ holds true for all $\kappa_2\ge \frac{3}{\log(2m)-\log(m+1)}$, where the denominator is bounded in a narrow range $\log(5/3) \le \log(2m)-\log(m+1) < 1$ since the recursive method is only effective for $m > 4$.

Precise gate and ancillary counts may be different depending on how we group the input states at each recursive step. For example, five input states can be grouped into $4+1$ or $3+2$ states, each leading to a pairing circuit of five ancillaries and CSWAP gates. On the other hand, for six inputs, grouping into $4+2$ and $3+3$ states requires six CSWAP gates with six ancillaries and seven CSWAP gates with five ancillaries, respectively.
It is important to stress that in our construction, we obtain $O\ln(m)$ measurements because each $\mathcal{U}_{m/2}$ block shares the same ancillaries.

\subsection{Swap test for Arbitrary Number of States}
\label{sec: th}
In this section, we use our pairing unitary $\mathcal{U}_{m}$ to complete the circuit that extends the SWAP test in Figure \ref{fig: circuit1} to any $m\geq 2$.

\begin{thm}
\label{thm: 1}
Given $m$ quantum states $\ket{\phi_{1}}, ..., \ket{\phi_{m}}$ ($m\geq 2)$, there exists a circuit that is equivalent in terms of measurements to the following mapping:
\begin{equation}\nonumber
   \ket{0}\displaystyle\otimes\ket{0}^{d_{m}}\displaystyle\otimes_{i=1}^{m}\ket{\phi_{i}} \rightarrow \displaystyle\sum_{ij=1}^{n}\frac{1}{2^{(d_{m}+1)/2}}\left[\sqrt{p_{0ij}}\ket{0}\ket{ij}\ket{y_{ij}} + \sqrt{1-p_{0ij}}\ket{1}\ket{ij}\ket{y'_{ij}}\right],
\end{equation} 
Where $\ket{y_{ij}}$ and $\ket{y'_{ij}}$ represent the entire state registers, $\ket{ij}$ denotes a basis state of $d_{m}$ ancillaries used in $\mathcal{U}_{m}$. The total ancillary count is $d_{m} +1$.  The probability $p_{0ij}$ of measuring the ancillary state $\ket{0ij}$ directly relates to the overlap of the state pair:
\begin{equation}\nonumber
p_{0ij}=\frac{1 + |\langle \phi_{i}|\phi_{j}\rangle|^{2}}{2^{d_{m}}}
\end{equation}
Moreover, the circuit uses $d_{m}=O(\ln m)$ ancillaries and  $c_{m}=O(m)$ CSWAP gates.
\end{thm}

\begin{figure}
    \centering
    \tikzset{every picture/.style={scale=3}}
    \input{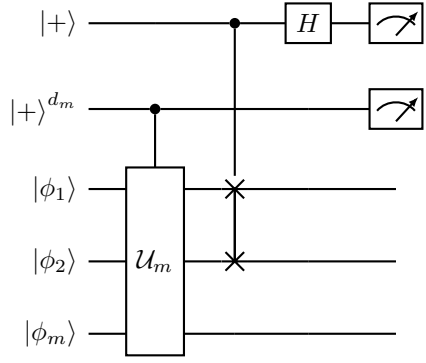}
    \caption{SWAP test for $m\geq 2$ quantum states. The circuit implements a 2-state SWAP test (see Figure \ref{fig: circuit1}) on the superposition state generated by the pairing unitary $\mathcal{U}_{m}$. }
     \label{circuit: generic}
\end{figure}

\begin{proof}
We construct a circuit as in Figure \ref{circuit: generic}.  Using the notations of Lemma \ref{lem: 1}, the last CSWAP and Hadamard gates transform the ancillaries and the state registers.
\begin{equation}
\label{eq: gen}
   \ket{0}\displaystyle\otimes\ket{0}^{d_{m}}\displaystyle\otimes_{i=1}^{m}\ket{\phi_{i}} \rightarrow \displaystyle\sum_{ij=1}^{n}\frac{1}{2^{(d_{m}+2)/2}}\left[\ket{0}\ket{ij}\left(\ket{\phi_{i}}\ket{\phi_{j}} + \ket{\phi_{j}}\ket{\phi_{i}}\right) +  \ket{1}\ket{ij}\left(\ket{\phi_{i}}\ket{\phi_{j}} - \ket{\phi_{j}}\ket{\phi_{i}}\right)\right]\ket{G_{ij}}.
\end{equation} 
This state is equivalent to the compact form in Theorem \ref{thm: 1} for the purpose of measurement.
\end{proof}

\subsection{Sample Complexity}



Given $m$ inputs, it is possible that by using $O(m)$ measurements, we could construct a pairing unitary $\mathcal{U}$ with fewer CSWAP gates. However, $O(m)$ measurements requires an exponential number of classical registers to store the results. Moreover, in this section, we argue that a polynomial sample complexity constraints the number of measurements. We define sample complexity as the number of copies of the $m$ quantum states needed to guarantee that our estimate of all state overlaps is within a precision $\epsilon>0$.

We assume that we have access to an oracle $\mathcal{O}_{m}$ that generates $m$ quantum states $\ket{\phi_{1}}$,..., $\ket{\phi_{m}}$ in time $T$ as inputs to the SWAP test circuit (Figure \ref{circuit: generic}). We can call $\mathcal{O}_{m}$ to estimate $n=m(m-1)/2$ state overlaps $\hat{\bm{\delta}}=\{\delta_{ij}\}_{ij=1}^{n}$ . The sample complexity is defined as the number of calls $N$ needed so that the estimates' error compared to the ground truth overlaps $\bm{\delta}=\{|\langle \phi_{i}|\phi_{j}\rangle|^{2}\}_{iy=1}^{n}$ is less than $\epsilon >0$. We evaluate the error in the estimate as $E(||\bm{\delta} -\widehat{\bm{\delta}}||_{2})$, where
\begin{equation}
    ||\bm{\delta} -\widehat{\bm{\delta}}||_{2}^{2} =\displaystyle\sum_{i< j}(\delta_{ij}-\hat{\delta}_{ij})^{2}.
\end{equation}
We show that an $\epsilon$-precision estimate requires at most $O\left(\frac{2^{2d_{m}}}{\epsilon^{2}}\right)$ calls to the oracle and $O\left(\frac{m2^{2d_{m}}}{\epsilon^{2}}\right)$ CSWAP gates.

\begin{thm}
\label{thm: 2}
The algorithm from Theorem \ref{thm: 1} needs at most $O\left(\frac{2^{2d_{m}}}{\epsilon^{2}}\right)$ calls to the oracle $\mathcal{O}_{m}$ to obtain an estimate $\hat{\delta}$ of $\delta=\{|\langle \phi_{i}|\phi_{j}\rangle|^{2}\}_{i< j=1}^{m}$ such that $E(||\bm{\delta} -\widehat{\bm{\delta}}||_{2})\leq \epsilon$. 
\end{thm}
\begin{proof}
Denote $\bm{p_{0}}=\{p_{0ij}\}_{ij=1}^{n}$ where $p_{0ij}$ is defined in Theorem \ref{thm: 1}. $\bm{\delta} = 2^{d_{m}}\bm{p_{0}} -1$. Denote $\widehat{\bm{p_{0}}}$ the maximum-likelihood estimate of $\bm{p_{0}}$ using $N$ copies of the $m$ quantum states. Then,
\begin{equation}\nonumber
    \begin{split}
        E(||\bm{\delta} - \bm{\hat{\delta}}||_{2}^{2}) & = 2^{2d_{m}}\displaystyle\sum_{ij=1}^{n}E(p_{0ij} - \widehat{p_{0ij}})^{2} \\
        & \overset{(a)}{=} 2^{2d_{m}}\displaystyle\sum_{ij=1}^{n}var(\widehat{p_{0ij}}) \\
        & = 2^{2d_{m}}\displaystyle\sum_{ij=1}^{n}\frac{p_{0ij}(1-p_{0ij})}{N} \\
        & \overset{(b)}{\leq} \frac{2^{2d_{m}}}{N}\displaystyle\sum_{ij=1}^{2^{d_{m}}}p_{0ij}(1-p_{0ij})) \\
        & \leq \frac{2^{2d_{m}}}{N}\left(1 - \frac{1}{2^{d_{m}}}\right) \\
        & \leq \frac{2^{2d_{m}}}{N},
    \end{split}
\end{equation}
where $(a)$ uses the fact the maximum-likelihood estimator $\widehat{\bm{p_{0}}}$ is unbiased; and, $(b)$ that $n \leq 2^{d_{m}}$. The result follows from the concavity of the square root function. 
\end{proof}

According to Theorem \ref{thm: 2}, since $d_{m}=O(\ln m)$, we need at most a polynomial number of calls to the oracle $\mathcal{O}_{m}$ for $\hat{\delta}$ to be within $\epsilon$ of $\delta$ in expectation. The next result shows that any algorithm with a pairing unitary $\mathcal{U}$ using $D_{m}$ measurements and followed by a standard SWAP test as in Figure \ref{circuit: generic} requires at least $\Omega\left(\frac{2^{2D_{m}}}{\epsilon^{2}}\right)$ samples for  $E(||\bm{\delta} -\widehat{\bm{\delta}}||_{2})\leq \epsilon$.

\begin{thm}
\label{thm: 3}
Consider a pairing unitary $\mathcal{U}$ followed by a SWAP test to compute all state overlaps $\bm{\delta}$ between $m$ quantum states. If the pairing unitary uses $D_{m}$ measurements, then there exists a state $\ket{\phi_{1}}, ..., \ket{\phi_{m}}$  such that at least $\Omega\left(\frac{2^{2D_{m}}}{\epsilon^{2}}\right)$ samples are needed for  $E(||\bm{\delta} -\widehat{\bm{\delta}}||_{2})\leq \epsilon$.
\end{thm}

\begin{proof}
We use results on the minimax rate for multinomial probabilities (see \cite{steinhaus1957problem}, \cite{braess2004asymptotic}) that implies that $\inf_{\hat{\delta}}\sup_{\ket{\phi_{1}}, ...\ket{\phi_{m}}}E(||\delta -\hat{\delta}||_{2}^{2}\geq \Omega\left(2^{2D_{m}}\frac{1}{N}\right)$. 
\end{proof}

Theorem \ref{thm: 3} implies that for a pairing unitary followed by a SWAP test to use at most a polynomial number of calls to the oracle $\mathcal{O}_{m}$, it is necessary to limit the number of measurements to $O(\ln m)$. 

\section{Numerical experiments}
\label{sec: num}

\begin{table}[h!]
\centering
\begin{center}
 \begin{tabular}{|c|c|} 
 \hline
 Number of CSWAP & Success rate  \\ 
 \hline
 11& 0.90 \\ 
 \hline
 10 & 0.79  \\
 \hline
\rowcolor{Gainsboro!60} 9& 0.47 \\
 \hline
 8 & 0.0 \\ 
 \hline
\end{tabular}
\end{center}
\label{tab: ga1}
\caption{Probability of the genetic algorithm to find a circuit with 8 inputs among all circuits using six ancillaries but a differing number of CSWAP gates. This table shows that the genetic procedure does not find pairing unitary circuits with six ancillaries but fewer CSWAP gates than the one proposed in Lemma \ref{lem: 1} (shaded row). Success rates are computed as the fraction of genetic trials that find a circuit capable of labelling all 28 overlaps with one the 64 base states encoded by the six ancillaries. The genetic algorithm parameters are as follows: population size is $10^{6}$, mutation rate is $0.5$ and number of iterations is $20,000$.}
\label{table:2}
\end{table}

In this section, we develop an evolutionary optimization procedure to find empirically for given a number of inputs $m$ the maximum number of overlaps that a pairing unitary circuit can label with fixed numbers of control swap $c_{m}$ and ancillaries $d_{m}$. The objective is to test empirically whether we can find pairing unitary circuits with fewer control swaps than the pairing unitary proposed in section 2, while maintaining the same number of ancillaries. Formally, if we denote $\mathcal{F}(m, c_{m}, d_{m})$ the sets of all circuits with $c_{m}$ control swap gates and $d_{m}$ ancillaries, we want to solve the following optimization problem:
\begin{equation}
\label{eq: opt}
    \min_{f\in \mathcal{F}(m, c_{m}, d_{m})} \frac{m(m-1)}{2} - G(f),
\end{equation}
where $G(f)$ is the number of pairs $i<j$ for which on any input of size $m$, the circuit $f$ outputs a state of the form $\ket{\phi_{i}, \phi_{j}*}$ or $\ket{\phi_{j}, \phi_{i}*}$.

To solve the minimization problem \eqref{eq: opt}, we use a genetic algorithm \cite{Hart, Yu, Potocek} \footnote{Code is available at the following \href{https://github.com/Gitiauxx/quant_genetic_algorithm}{github repo}.}. A circuit is a sequence of CSWAP gates $cs(a, t_{1}, t_{2})$, where $a\in\{1, ..., d_{m}\}$ is the control ancillary and $t_{1}, t_{2}\in \{1, ..., m\}$ are the two target registers. Therefore, $|\mathcal{F}(m, c_{m}, d_{m})|= [d_{m}m(m-1)]^{c_{m}}$.

Genetic algorithms are evolutionary optimization procedures that given an initial population of circuits $\{f_{1}, ..., f_{L}\}$ evolves it by :
\begin{itemize}
    \item \textbf{Selection} of the two parent circuits with the lowest cost function $m(m-1)/2 - G(f)$;
    \item \textbf{Crossover} of the two parents circuits along a randomly selected pivot. That is, from parent circuits $f_{pa}=(cs_{1}...cs_{c_{m}})$ and $f^{'}_{pa}=(cs_{1}^{'}...cs_{c_{m}}^{'})$, we generate two offspring as $f_{off}=(cs_{1}, ..., cs_{i-1}, cs^{'}_{i}, ...cs_{c_{m}}^{'})$ and $f_{off}^{'}=(cs_{1}^{'}, ..., cs_{i-1}^{'}, cs_{i}, ...cs_{c_{m}})$, where $i$ is uniformly drawn from $\{1, ..., c_{m}\}$.
    \item \textbf{Mutation} of the offspring $f_{off}$ and $f_{off}^{'}$. That is, with probability $p_{mutation}$, one CSWAP gate of $f_{off}$ is mutated to $cs(a, t_{1}, t_{2})$, where $a$ is uniformly drawn from $\{1, ..., d_{m}\}$ and $t_{1}, t_{2}$ are uniformly drawn from $\{1, ..., c_{m}\}$.
    \item \textbf{Replacement} of two randomly chosen members of the initial population by the new offspring $f_{off}$ and $f_{off}^{'}$.
\end{itemize}
The genetic algorithm iterates over this selection-crossover-mutation-replacement procedure. Genetic algorithms are guaranteed to converge to a global minimum provided a large number of iterations $M$ is allowed \cite{genetic}. The procedure is parameterized by its population size $L$, its mutation rate $p_{mutation}$ and its number of iteration $M$.

In Table \ref{tab: ga1}, we run the genetic algorithm $100$ times with different seeds to minimize the cost \eqref{eq: opt} with $m=8$ inputs. We repeat the experiment for different number of CSWAP gates and report the success rate of the genetic algorithm. We define success rate as the fraction of simulations for which the genetic algorithm achieves a zero cost function out of the 100 seeds we use to initialize the algorithm. For $m=8$, our construction in section 2 proposes a circuit with $c_{m}=9$ control swap gates and $d_{m}=6$ ancillaries. In Table \ref{tab: ga1}, we observe that the genetic algorithm only finds $47\%$ of the time a circuit of size $c_{m}=9$ and $d_{m}=6$ that creates an output state $\ket{\phi_{i}, \phi_{j}*}$ or $\ket{\phi_{j}, \phi_{i}*}$ with positive probability for all $i<j$. We posit that the size of the search space $\mathcal{F}(m, c_{m}, d_{m})$ ($\approx 10^{23}$) limits the ability of the genetic algorithm to converge to the global optimum. Moreover, the genetic algorithm does not appear to find a solution with fewer than nine CSWAP gates. Although this failure could result from shortcomings of our genetic algorithm procedure, it also may indicate that empirically, standard optimization procedures do not seem to find circuits with $c_{m} < 9$ while maintaining $d_{m}\leq 6$.

Figure \ref{fig:pop} shows for $m=8$ the average number of overlaps labelled by circuit that are solutions of the genetic algorithm search. Any number of overlaps less than $28$ indicates a pairing unitary circuit that fails to account for all overlaps. Figure \ref{fig:pop} confirms the  main results in Table \ref{tab: ga1}. First, regardless of the genetic algorithm population size, circuits with fewer than nine CSWAP gates never manage to encode the $28$ overlaps: on average they miss $2.5$ to $2$ overlaps. Secondly, the search for a solution of \eqref{eq: opt} is challenging in the sense that even with nine CSWAP, where we know that a solution exists (see Lemma \ref{lem: 1}), the search generates circuits that on average miss $0.5$ to $1$ overlap, depending on the genetic population size. 

\begin{figure}
    \centering
    \includegraphics[scale=0.5]{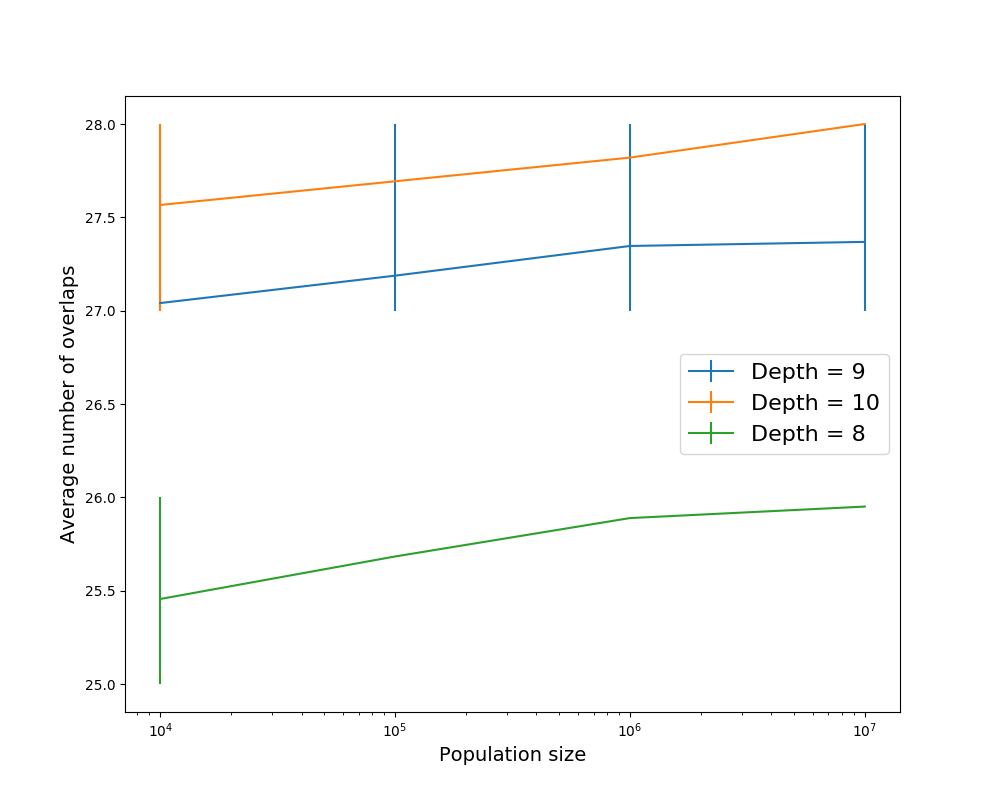}
    \caption{Effect of population size of the number of inner products $\ket{\phi_{i}, \phi_{j}}$ on the first two registers of circuits found by the genetic algorithm procedure. The number of inputs is $m=8$; the number of ancillaries is constrained to be at most $6$. The maximum of possible overlaps is $28$. The plot shows an average over 100 simulations, along with the $25\%$ and $75\%$ quantiles, of the number of overlaps labelled by circuits found by the genetic algorithm.}
    \label{fig:pop}
\end{figure}

\section{Discussion}
\label{sec: discussion}

We have expanded the simple SWAP test to calculate the overlap between all pairs of $m$ arbitrary input states. At the core of our algorithm is a pairing unitary $\mathcal{U}_{m}$ consisting of CSWAP gates applied to the state registers and controlled by ancillary qubits. The pairing unitary prepares a superposition of every state pair appearing in the first two state registers and matched with a labelling state provided by the ancillary basis. Since  $\mathcal{U}_{2}$ is trivial and $\mathcal{U}_{3}$ with two ancillary and CSWAP gates is obvious, we have focused on the construction of a circuit for $\mathcal{U}_{4}$ that is  proven to be optimal using three CSWAP gates and ancillaries. The circuit complexity is analyzed to ensure minimization of the CSWAP and ancillary counts. A $k-2$-step recursive method expands $\mathcal{U}_{4}$ to $\mathcal{U}_{m}$ for any $m=2^{k}$ where the exact gate and ancillary counts have been derived as $3m/2-3$ and $3\log{(m/2)}$, respectively. Such pairing unitary can be applied to any arbitrary number $m$ for $2^{k-1}<m\le{2^{k}}$ through padding to $2^{k}$ states. Alternatively, the pairing unitary for any $m$ can be constructed by decomposing an integer $m$ into a sequence of $2^{k}$. A more general recursive construction has also been developed by partition down the middle, i.e. into two groups of $\ceil{m/2}$ and $\ceil{m/2}-1$ input states at every step. This allows more flexibility with the basic building block of $\mathcal{U}_{2}$, $\mathcal{U}_{3}$, and $\mathcal{U}_{4}$. All recursive constructions discussed above have been demonstrated to have CSWAP and ancillary complexities of $O(m)$ and $O(\log(m))$, respectively. The exact gate and ancillary counts might be of interest in some scenarios for future study.

Our algorithm could serve as an intermediate step in quantum machine learning algorithms, e.g. quantum support vector machine or quantum clustering, that rely on inner products. Existing work in quantum machine learning assumes a QRAM oracle \cite{giovannetti2008quantum} that uses $O(\log m)$ quantum states to assign an address to each classical vector using (e.g. \cite{rebentrost2014quantum}, \cite{kerenidis2019q}).  We only assume that each classical vector is encoded into $O(\log q)$ qubits. Given the challenges of implementing QRAM \cite{ciliberto2018quantum}, we believe that our approach is a step toward a practical implementation of machine learning routines on quantum systems, capable of reducing the complexity from $O(q)$ on a classical system to $O(\log q)$ on a quantum system.

Recently, in \cite{galvao2020quantum}, methods for estimating coherence witness and dimension witness have been proposed that would directly benefit from the approach presented herein.  Generalizations of quantum fingerprinting \cite{buhrman2001quantum} and photon distinguishability \cite{giordani2019experimental} to higher dimensions are natural and can be explored in this context as well. Recent work on determining quantum entanglement presents another interesting research direction \cite{Foulds2021}, as does the possibility to extend the method of inner-product based genomic classifiers given in \cite{kathuria2020implementation}. 
In all of the above mentioned applications,  processing all states in the same circuit is beneficial due to its ability to generate true randomness in state sampling, which offers the possibility to manage the measurement and reduce the sampling complexity. Obtaining all state overlaps by pairwise comparison based on 2-state SWAP test requires additional circuits to randomly select the pair to be compared. In contrast, in our approach, we do not need any additional quantum circuits to achieve this randomness. 

\section{Acknowledgements}
The authors gratefully acknowledge financial support from the George Mason University Quantum Science and Engineering Center and productive discussions with Michael Jarret, Andrew Glaudell and Ernesto Galv{\~a}o. 

\bibliographystyle{unsrt}
\bibliography{references}

\end{document}